\documentclass[letterpaper, 10 pt, conference]{ieeeconf}  % Comment this line out if you need a4paper

\IEEEoverridecommandlockouts                              % This command is only needed if
\usepackage{amsmath}
\usepackage{amsfonts}
\usepackage{amssymb}
\usepackage{diagram}
\usepackage{bm}
\usepackage{mathrsfs}
\usepackage{graphicx}
\usepackage{subfigure}
\usepackage{wrapfig}
\usepackage{algorithm}
\usepackage{algpseudocode}
\usepackage{color}
\usepackage{listings}
\allowdisplaybreaks[4]
\newenvironment{breakablealgorithm}
{% \begin{breakablealgorithm}
		\begin{center}
			\refstepcounter{algorithm}% New algorithm
			\hrule height.8pt depth0pt \kern2pt% \@fs@pre for \@fs@ruled
			\renewcommand{\caption}[2][\relax]{% Make a new \caption
				{\raggedright\textbf{Algorithm~\thealgorithm} ##2\par}%
				\ifx\relax##1\relax % #1 is \relax
				\addcontentsline{loa}{algorithm}{\protect\numberline{\thealgorithm}##2}%
				\else % #1 is not \relax
				\addcontentsline{loa}{algorithm}{\protect\numberline{\thealgorithm}##1}%
				\fi
				\kern2pt\hrule\kern2pt
			}
		}{% \end{breakablealgorithm}
		\kern2pt\hrule\relax% \@fs@post for \@fs@ruled
	\end{center}
}
\makeatother

\DeclareMathOperator{\diag}{diag}
\def\({\left(}
\def\){\right)}

\newtheorem{assumption}{Assumption}
\newtheorem{thm}{Theorem}%[section]
\newtheorem{rema}{Remark}%[section]
\begin{document}

\title{\LARGE \bf
 Quantized Distributed Estimation with Event-triggered Communication and Packet Loss
}

\author{Ying Wang$^{1,2}$, Yanlong Zhao$^{2,3}$, Ji-Feng Zhang$^{4,2,3}$ and Karl Henrik Johansson$^{1}$% <-this % stops a space
\thanks{*This work is supported in part by Swedish Research Council Distinguished Professor Grant 2017-01078, Knut and Alice Wallenberg Foundation Wallenberg Scholar Grant, and the Swedish Strategic Research Foundation FUSS SUCCESS Grant. This work is supported in part by National Natural Science Foundation of China under Grants 62025306, 62433020, 62303452, T2293770 and 62473040, and CAS Project for Young Scientists in Basic Research under Grant YSBR-008.}% <-this % stops a space
\thanks{$^{1}$Division of Decision and Control Systems, KTH Royal Institute of Technology, Stockholm 11428, Sweden. {\tt\small wangying96@amss.ac.cn}, {\tt\small kallej@kth.se}}%
\thanks{$^{2}$ The State Key Laboratory of Mathematical
Sciences, Institute of Systems Science, Academy of Mathematics and Systems Science, Chinese Academy of Sciences, Beijing 100190, P. R. China. {\tt\small ylzhao@amss.ac.cn}}%
\thanks{$^{3}$School of Mathematical Sciences, University of Chinese Academy of Sciences, Beijing 100049, P. R. China.}%
\thanks{$^{4}$School of Automation and Electrical Engineering, Zhongyuan University of Technology, Zhengzhou 450007, P. R. China
          {\tt\small jif@iss.ac.cn} }%
}

\maketitle
\thispagestyle{empty}
\pagestyle{empty}

%%%%%%%%%%%%%%%%%%%%%%%%%%%%%%%%%%%%%%%%%%%%%%%%%%%%%%%%%%%%%%%%%%%%%%%%%%%%%%%%
\begin{abstract}
This paper focuses on the problem of quantized distributed estimation with event-triggered communication and packet loss, aiming to reduce the number of transmitted bits.
The main challenge lies in the inability to differentiate between an untriggered event and a packet loss occurrence.
This paper proposes an event-triggered distributed estimation algorithm with quantized communication and quantized measurement,
in which it introduces a one-bit information reconstruction method to deal with packet loss.
The almost sure convergence and convergence rate of the proposed algorithm are established.
Besides, it is demonstrated that the global average communication bit-rate decreases to zero over time.
Moreover, the trade-off between communication rate and convergence rate is revealed,  providing guidance for designing the communication rate required to achieve the algorithm's convergence rate.
A numerical example is supplied to validate the findings.
\end{abstract}

%%%%%%%%%%%%%%%%%%%%%%%%%%%%%%%%%%%%%%%%%%%%%%%%%%%%%%%%%%%%%%%%%%%%%%%%%%%%%%%%
\section{INTRODUCTION}
Wireless sensor networks (WSNs) are widely utilized in fields such as industry (process automation monitoring), environment (fire detection), traffic (traffic conditions detection) due to their ease of deployment, flexibility and reliability \cite{ASSC2002,LWHCG2025,GHZDY2020}.
However, to maintain these advantages, WSNs are typically battery-powered, which imposes energy and communication constraints \cite{GHZDY2020,RSCB2014}. Consequently, the data they transmit and measure are often quantized data with finite bits.
Therefore, quantized distributed estimation over sensor networks has become a critical theoretical issue in WSNs research and has attracted numerous attention \cite{DKMRS2010,DWDG2017,FCZ2022,KMR2012,ZCGXLJ2018,WGZZ2024}.
Existing works on quantized distributed estimation generally falls into three categories. The first is distributed estimation with quantized measurements and accurate communication \cite{DWDG2017,FCZ2022}.  The second is distributed estimation with accurate measurements and quantized communication \cite{KMR2012,ZCGXLJ2018}.  The third is distributed estimation with quantized measurements and quantized communication \cite{WGZZ2024}.

It is well-known that data transmission consumes more energy than local data processing in WSNs \cite{PK2020}. To conserve energy of WSNs, data transmission must be minimized as much as possible.
Event-triggered mechanism is a highly effective approach for reducing data transmission and finds broad applications in areas such as consensus control \cite{DFJ2011, LWZ2024}, adaptive tracking \cite{XZZ2025}, distributed estimation \cite{HXWJ2022, KLZZ2024}, and distributed optimization \cite{KHT2018, YXHC2024}, among others.
For example, an event-triggered distributed estimation algorithm is proposed under accurate communication, and its communication rate is shown to decay to zero \cite{HXWJ2022}. Furthermore, it is extended to quantized communication cases and is based only on the decaying communication bit rate \cite{KLZZ2024}.

Packet loss is a common communication uncertainty in WSNs, often resulting from channel interference, network congestion, and environmental factors \cite{SAS2006,ACP2019,WWXZ2025}. It not only leads to data loss but also significantly impacts system performance.
This paper will explore the problem of distributed estimation with quantized measurements and quantized communication under an event-triggered mechanism in the presence of packet loss. Compared with the existing works \cite{HXWJ2022,KLZZ2024}, the primary challenge of this paper lies in the sensor's inability to discern whether the absence of received data is due to an untriggered event or packet loss.
To address this, the paper proposes an effective information compensation method, enabling quantized distributed estimation with decaying communication bits, even in scenarios involving packet loss.
The main contributions are as follows.
\begin{itemize}
\item  An event-triggered quantized distributed estimation algorithm is proposed to achieve distributed estimation under packet loss, which only depends on the decaying bit-rate communication. In particular, a stochastic event-triggered communication mechanism with a Laplace dither and a growing threshold is applied to guarantee the decaying communication bit rate, and a one-bit information reconstruction method to compensate the loss caused by packet loss in the expected sense.

\item The convergence properties and communication rate of the proposed algorithm are established. Specifically, this paper establishes the almost sure convergence under cooperative excitation condition, with convergence rate close to $O(\sqrt{\frac{\ln k}{k^{1-\nu}}})$. Besides, the global average communication bit-rate is shown as $O(\frac{1}{k^{\nu}})$, decreasing to zero over time. In addition, this paper reveals the quantitative relationship that constrains communication rate and convergence rate, guiding the design of communication rate to achieve required convergence rate.
\end{itemize}

The remainder of this paper is organized as follows. Section II provides  problem formulation. Section III introduces the algorithm design. The main results are presented in Section IV,  which  includes convergence and communication rate.  Section V gives a simulation example. Section VI summarizes the conclusion and provides future directions.

\section{Problem formulation}

\subsection{Graph theory}

In order to describe the relationship between sensors, an undirected graph $\mathcal{G}=(\mathcal{V},\mathcal{E},\mathcal{A})$ is introduced here, where $\mathcal{V}=\{1,2,\ldots,m\}$ is the set of sensors and $\mathcal{E}\in\mathcal{V}\times\mathcal{V}$ is the edge set describing the communication between sensors.
The adjacency matrix $\mathcal{A}=\{a_{ij}\}\in\mathbb{R}^{m\times m}$ describes the structure of the graph $\mathcal{G}$, where $a_{ij}>0$ if $(j,i)\in \mathcal{E}$, and $a_{ij}=0$, otherwise. The adjacency matrix $\mathcal{A}$ satisfies $a_{ij}=a_{ji}$ for all $i,j=1,\ldots,m$.
Besides, the set of the neighbors of sensor $i$ is denoted as $\mathcal{N}_i=\{j\in \mathcal{V}|(i,j)\in\mathcal{E}\}$. The laplacian matrix of $\mathcal{G}$ is defined as $\mathcal{L}=\sum_{i=1}^ma_{ij}-\mathcal{A}$.

\subsection{System description}

Consider a multi-agent network consisting of $m$ sensors, the dynamic of the $i$-th $(i=1,\ldots,m)$ sensor at time $k$ is a linear stochastic model with quantized measurements as:
\begin{equation}\label{M}
\begin{cases}
y_{k,i}=\phi_{k,i}^T\theta+d_{k,i},\\
s_{k,i}=I_{\{y_{k,i}\leq C_i\}},
\end{cases}
\end{equation}
where $\phi_{k,i}\in \mathbb{R}^n$ is an $n$-dimensional regressor, $d_{k,i}$ is a stochastic noise and $\theta\in\mathbb{R}^n$ is an unknown parameter vector to be estimate. And $y_{k,i}$ is a scalar measurement of sensor $i$, which can not be exactly measured but be measured by a quantized measurement $s_{k,i}$, where $C_i$ is its fixed threshold.

In order to save communication cost, one-bit encoder is to be designed to determine how the estimate $\hat\theta_{k-1,i}$ of the sensor $i$ is transmitted to its neighbor $j\in\mathcal{N}_i$ by virtue of one-bit data $z_{k,i}$ as follows.
\begin{equation}\label{en_z}
z_{k,i}=Q(\hat\theta_{k-1,i}),
\end{equation}
where $Q(\cdot):\mathbb{R}^n\rightarrow\mathbb{R}$ is a compression coding function, which could transform an vector to a one bit data.

Then, to further reduce communication traffic, an event-triggered mechanism $\gamma_{k,i}^e$ is used to determine whether the one-bit encoding data $z_{k,i}$ of the sensor $i$ should be transmitted to its neighbors $j\in\mathcal{N}_i$.

During the transmission of $z_{k,i}$,  the data packet loss on the channel $(i,j)\in\mathcal{E}$ is expressed as
\begin{equation}\label{pl_gamma}
\gamma_{k,ji}^d=
\begin{cases}
0, & \text { data packet loss on the channel } (i,j);  \\
1, & \text { otherwise}.
\end{cases}
\end{equation}
Thus, the indicator whether the neighbor sensor $j$ receives one-bit data $z_{k,i}$ from sensor $i$ is given by $\gamma_{k,ji}=\gamma_{k,ji}^d \gamma_{k,i}^e$. But the neighbor $j$ cannot distinguish between $\gamma_{k,ji}^d$ and $\gamma_{k,i}^e$, i.e., the data received by the sensor $j$ is $\left\{\gamma_{k,ji}, \gamma_{k,ji} z_{k,i}\right\}$.

\subsection{Assumptions}
To proceed our analysis, we introduce some assumptions about the graph, the regressor, the noise and packet loss.

\begin{assumption}\label{AG}
The graph $\mathcal{G}$ is connected.
\end{assumption}

\begin{assumption}\label{AR}
\textbf{(Cooperative excitation condition)}
The regressor $\{\phi_{k,i}\}$ satisfy $\|\phi_{k,i}\| \leq \bar{\phi}<\infty$ and there exist a positive integer $h$ and a positive number $\delta_{\phi}$ such that
\vspace{-5pt}
\[
\sum^{m}_{i=1} \sum^{k+h-1}_{l=k}\phi_{l,i}\phi_{l,i}^T\geq\delta_{\phi}^2I_{n\times n}, \quad \forall k\geq 1.
\]
\end{assumption}

\begin{assumption}\label{AP}
Let $\Omega\subset\mathbb{R}^n$ be a known bounded, convex, and compact set such that $\theta\in \Omega$, and let $\bar{\theta}=\sup_{\eta\in\Omega}\|\eta\|.$
\end{assumption}

\begin{assumption}\label{AD}
The conditional distribution function of the noise $d_{k,i}$ given $\mathcal{F}_{k-1}$ is denoted by $F_{k,i}(\cdot)$,  and its condition density function $f_{k,i}(x)=\frac{dF_{k,i}(x)}{dx}$ satisfies
\vspace{-5pt}
\[
\underline{f}=\inf_{k\geq 1}\min_{1\leq i\leq m}\min_{|C_i-x|\leq\bar{\phi}\bar{\theta}} f_{k,i}(x)>0,
\]
where  $\mathcal{F}_k=\{d_{l,i}, i=1,\ldots,m, 0\leq l\leq k\}$.
\end{assumption}

%\begin{rema}\label{RAR}
% The prior information condition of unknown parameter (i.e., Assumption \ref{AP}) is a common assumption of quantized identification studies \cite{GZ2013,WZZG2022,WZZ2021,ZWZ2021,ZZG2022}, which is mainly used to guarantee the boundness of the estimate.
 %Compared with \cite{GZ2013,WZZG2022,WZZ2021,ZWZ2021,FCZ2022,You2015}, the condition of independently and identically distributed noise is no longer required in Assumption \ref{AD}.
 %The assumptions about prior information of unknown parameter and noises are common in the previous quantized identification study\cite{WZZG2022}. Besides, there are many common noises conforming Assumption \ref{AD}, such as Gaussian noises, Laplace noises and so on.
%\end{rema}

\begin{assumption}\label{ADP}
For all indices $i$ and $j$, the packet loss sequence $\{\gamma_{k,ij}^d\}$ subjects to i.i.d. Bernoulli process with probability $p\in(0,1)$, i.e., $\mathbb{P}(\gamma_{k,ij}^d=0)=p$ and $\mathbb{P}(\gamma_{k,ij}^d=1)=1-p$. The packet loss $\{\gamma_{k,ij}^d\}$ is independent of the noise $\{d_{k,i}\}$ and the event-triggered mechanism $\{\gamma_{k,i}^e\}$.
\end{assumption}

The goal of this paper is to develop a distributed estimation algorithm with a one-bit encoder and an event-triggered communication mechanism to estimate the unknown parameter based on quantized measurement and quantized communication under packet loss.

\section{Algorithm Design}
This  section  focuses  on  designing  an event-triggered quantized distributed estimation algorithm that can deal with information loss caused by packet loss and achieve distributed estimation with as few transmission bits as possible.
The algorithm is constructed as follows.

\begin{breakablealgorithm}
\caption[]{} %ETQD algorithm
\label{Algorithm}
For any given sensor $i\in\{1,2,\ldots,m\}$, begin with initial estimates $\hat{\theta}_{0,i}\in\mathbb{R}^n$. The algorithm is defined as follows:
\begin{algorithmic}[1]
\State Design of one-bit encoder: based on the generated dither $\omega_{k,i}$ with Laplace distribution $Lap(0,1)$,  the sensor $i$ encodes its estimate $\hat\theta_{k-1,i}$ into one-bit  data as:
\begin{align}\label{Aen_z}
z_{k,i}&=Q(\hat\theta_{k,i-1})=1-2I_{\{\psi_{k}^T\hat\theta_{k-1,i}+\omega_{k,i}\leq0\}}, \nonumber\\
&=\left\{
   \begin{array}{ll}
     1, & \hbox{if } ~\psi_{k}^T\hat\theta_{k-1,i}+\omega_{k,i}>0; \\
     -1, & \hbox{if } ~ \psi_{k}^T\hat\theta_{k-1,i}+\omega_{k,i}\leq0,
   \end{array}
 \right.
\end{align}
where the coding rule $\{\psi_{k}\}$ satisfy $\|\psi_{k}\| \leq \bar{\psi}<\infty$ and there exists a positive constant $\delta_{\psi}$ such that $\frac{1}{h}\sum^{k+h-1}_{l=k}\psi_{l}\psi_{l}^T\geq\delta_{\psi}^2I_{n\times n}$ for all $k\geq 1$.

\State Design of event-triggered communication mechanism:
\begin{align}\label{Aga_e}
\gamma_{k,i}^e
=\left\{
   \begin{array}{ll}
1, & \hbox { if } ~|\psi_{k}^T\hat\theta_{k-1,i}+\omega_{k,i}|>\hat{C}_k;  \\
0, & \hbox { otherwise, }
   \end{array}
 \right.
\end{align}
where $\hat{C}_k=\nu\ln k$; when $\gamma_{k,i}^e=1$, then $z_{k,i}$ will be sent to its all neighboring sensor $j\in \mathcal{N}_{i}$. Otherwise, $z_{k,i}$ cannot be transmitted to its any neighboring sensor.

\State Reconstruction of one-bit neighboring information:
\begin{align}\label{Aest_z}
\hat{s}_{k,ij}=\frac{\gamma_{k,ij} z_{k,j}}{1-p},
\end{align}
where $\gamma_{k,ij}=\gamma_{k,ij}^d\gamma_{k,j}^e$ and $p$ is packet loss probability.

\State Fusion estimation:
\begin{align}\label{AE}
\hat{\theta}_{k,i}=\Pi_{\Omega}\big\{&\hat{\theta}_{k-1,i}+\frac{\beta\phi_{k,i}}{k} (\hat{F}_{k,i}-s_{k,i})\nonumber\\
&+\frac{\alpha}{k^{1-\nu}} \sum_{j\in \mathcal{N}_{i}}a_{ij} \psi_k(\hat{s}_{k,ij}-\hat{G}_{k,i})\big\},
\end{align}
where the projection $\Pi_{\Omega}(\cdot)$ is defined as $\Pi_{\Omega}(x)= \arg\min_{z\in\Omega}\|x-z\|$ for all $x\in\mathbb{R}^n$; $\beta>0$ and $\alpha>0$ are step coefficients; $\hat{F}_{k,i}=F_{k,i}(C_i-\phi^T_{k,i}\hat\theta_{k-1,i})$ and $\hat{G}_{k,i}=G_{k}(\psi_k^T \hat{\theta}_{k-1,i})=G(\psi_k^T \hat{\theta}_{k-1,i}-\hat{C}_k)-G(-\psi_k^T \hat{\theta}_{k-1,i}-\hat{C}_k)$, where $G(\cdot)$ is the  distribution function of Laplace distribution $Lap(0,1)$.
\end{algorithmic}
\end{breakablealgorithm}

\begin{rema}
The proposed algorithm is developed based on the consensus-type distributed stochastic approximation algorithm including local information processing term and global information consensus term, with two key highlights:
\begin{itemize}
\item Design of a one-bit communication protocol with event-triggered mechanism: The protocol (\ref{Aen_z})-(\ref{Aga_e}) transforms an accurate estimate into data less than one-bit, inspired by \cite{KLZZ2024}. Specifically, the one-bit encoding protocol is designed using stochastic dither and a linear encoding rule that satisfies the persistent excitation condition. These elements jointly ensure the accurate restoration of the directional information of the estimate. The event-triggered mechanism incorporates Laplace dither and a growing threshold, effectively reducing the probability of event triggering to zero over time.
\item Construction of a one-bit information reconstruction method: The reconstruction method (\ref{Aest_z}) compensates for information loss caused by packet loss in the expected sense. This is feasible because the distributed algorithm primarily relies on directional information about the difference between the local estimate $\hat\theta_{k,i}$ and its neighboring estimate $\hat\theta_{k,j}$ to design the global consensus term. Specifically, $\hat{s}_{k,ij}-\hat{G}_{k,i}$ provides the directional information in the expected sense, as derived from $\mathbb{E}[\hat{s}_{k,ij}-\hat{G}_{k,i}]=\mathbb{E}{z}_{k,j}-\hat{G}_{k,i}= G_{k}(\psi_k^T \hat{\theta}_{k-1,j})-G_{k}(\psi_k^T \hat{\theta}_{k-1,i})$. This is further supported by the monotonicity of $G_{k}(x)$ and the packet loss property.
\end{itemize}
\end{rema}
\begin{rema}
The coding coefficients $\psi_k$ are designed to satisfy the persistent excitation condition \cite{S_Guo2020}, ensuring recoverability of the original vector from quantized data. A simple choice is to cyclically select $\psi_k$ from the standard basis vectors $\{e_1, \ldots, e_n\} \subset \mathbb{R}^n$, where $e_i$ denotes the $i$th unit vector.
\end{rema}

\section{Main results}
This section is concern with the main conclusions of this paper, including the almost sure convergence properties and the communication rate of the proposed algorithm. And then, we will give an analysis on the trade-off between the convergence rate and the communication rate.

\begin{thm}\label{thm_ac}
If Assumptions \ref{AG}-\ref{ADP} hold, then the proposed algorithm is almost surely convergent.
Moreover, for any $\tau>0$, its almost sure convergence rate  is
\[
\|\tilde{\theta}_{k,i}\|=O\left(\sqrt{\frac{(\ln k)^{1+\tau}}{k^{1-\nu}}}\right), \text{a.s.},
\]
if $2\sigma\geq 1-\nu$, where $\tilde{\theta}_{k,i}=\hat{\theta}_{k,i}-\theta$ is the estimation error, $\sigma=\frac{\alpha\beta h\lambda_2\underline{f}\underline{g}\delta_{\psi}^2 \delta_{\phi}^2}{2\beta\bar{\phi}^2\underline{f} +\alpha \underline{g}h\lambda_2\delta_{\psi}^2}$, $\lambda_2$ is the smallest nonzero eigenvalue of the Laplacian matrix $\mathcal{L}$, $\underline{g}=\inf_{k\geq 1}\min_{x\in[-\bar{\psi}\bar{\theta},\bar{\psi}\bar{\theta}]}k^{\nu}{g}_{k}(x)$,
${g}_{k}(x)=g(x-\hat{C}_k)+g(-x-\hat{C}_k)$, and $g(\cdot)$ is the density function of Laplace distribution $Lap(0,1)$.
\end{thm}
\begin{proof}
\textbf{Part I}:  Almost sure convergence.\\
From (\ref{AE}),  Assumptions \ref{AR}-\ref{AP}, Lemma 5.1 in \cite{WZZ2024} and the boundedness of $\hat{F}_{k,i}$ and $\hat{G}_{k,i}$, we have
\vspace{-5pt}
\begin{align}\label{thmmc1}
\tilde{\theta}_{k,i}^T\tilde{\theta}_{k,i}
\leq&\tilde{\theta}_{k-1,i}^T\tilde{\theta}_{k-1,i}+\frac{2\beta}{k} \tilde{\theta}_{k-1,i}^T\phi_{k,i} (\hat{F}_{k,i}-s_{k,i})\nonumber\\
&+\frac{2\alpha}{k^{1-\nu}}  \tilde{\theta}_{k-1,i}^T\sum_{j\in \mathcal{N}_{i}}a_{ij} \psi_k(\hat{s}_{k,ij}-\hat{G}_{k,i})\nonumber\\
&+\frac{2\beta\alpha}{k^{2-\nu}}(\hat{F}_{k,i}-s_{k,i})\phi_{k,i}^T\sum_{j\in \mathcal{N}_{i}}a_{ij} \psi_k(\hat{s}_{k,ij}-\hat{G}_{k,i})\nonumber\\
&+\frac{\beta^2}{k^2}\phi_{k,i}^T\phi_{k,i} (\hat{F}_{k,i}-s_{k,i})^2\nonumber\\
&+\frac{2\alpha^2 \psi_k^T\psi_k}{k^{2-2\nu}}\big(\sum_{j\in \mathcal{N}_{i}}a_{ij}(\hat{s}_{k,ij}-\hat{G}_{k,i})\big)^2\nonumber\\
\leq&\tilde{\theta}_{k-1,i}^T\tilde{\theta}_{k-1,i}+\frac{2\beta}{k} \tilde{\theta}_{k-1,i}^T\phi_{k,i} (\hat{F}_{k,i}-s_{k,i})+\frac{2\alpha}{k^{1-\nu}}\nonumber\\
& \cdot\tilde{\theta}_{k-1,i}^T\psi_k\sum_{j\in \mathcal{N}_{i}}a_{ij} (\hat{s}_{k,ij}-\hat{G}_{k,i})+O\(\frac{1}{k^{2-\nu}}\)\nonumber\\
&+\frac{2\alpha^2 \psi_k^T\psi_k}{k^{2-2\nu}}\big(\sum_{j\in \mathcal{N}_{i}}a_{ij}(\hat{s}_{k,ij}-\hat{G}_{k,i})\big)^2.
\end{align}
By the differential mean value theorem and Assumption \ref{AD},
\begin{align}\label{thmmc2}
&-\mathbb{E}[\tilde{\theta}_{k-1,i}^T\phi_{k,i} (\hat{F}_{k,i}-s_{k,i})|\mathcal{F}_{k-1}]\nonumber\\
%=&-\tilde{\theta}_{k-1,i}^T\phi_{k,i} \(\hat{F}_{k,i}-\mathbb{E}[s_{k,i}|\mathcal{F}_{k-1}]\)\nonumber\\
=&-\tilde{\theta}_{k-1,i}^T\phi_{k,i} \(F_{k,i}(C_i-\phi^T_{k,i}\hat\theta_{k-1,i}) -F_{k,i}(C_i-\phi^T_{k,i}\theta)\)\nonumber\\
=&-f_{k,i}\(\xi_{k,i}\)\tilde{\theta}_{k-1,i}^T\phi_{k,i}\phi^T_{k,i}\tilde\theta_{k-1,i}\nonumber\\
\leq& - \underline{f}\tilde{\theta}_{k-1,i}^T\phi_{k,i}\phi^T_{k,i}\tilde\theta_{k-1,i},
\end{align}
where $\xi_{k,i}$ is between $C_i-\phi^T_{k,i}\hat\theta_{k-1,i}$ and $C_i-\phi^T_{k,i}\theta$ such that $\hat F_{k,i} -F_{k,i}(C_i-\phi^T_{k,i}\theta)=-f_{k,i}\(\xi_{k,i}\)\phi^T_{k,i}\tilde\theta_{k-1,i}$.

Then, by Assumption \ref{ADP},  we have
\begin{align}\label{thmmc3}
&\mathbb{E}[\hat{s}_{k,ij}-\hat{G}_{k,i}|\mathcal{F}_{k-1}]=\mathbb{E}\left[\frac{\gamma_{k,ij}^d\gamma_{k,j}^e z_{k,j}}{1-p}\bigg|\mathcal{F}_{k-1}\right]-\hat{G}_{k,i}\nonumber\\
=&\mathbb{E}\left[\frac{\gamma_{k,ij}^d}{1-p}\right]\mathbb{E}\left[\gamma_{k,j}^e z_{k,j}|\mathcal{F}_{k-1}\right]-\hat{G}_{k,i}\nonumber\\
=&\hat{G}_{k,j}-\hat{G}_{k,i}
={g}_{k}(\zeta_{k,ij})(\tilde{\theta}_{k-1,j}^T\psi_k- \tilde{\theta}_{k-1,i}^T\psi_k).
\end{align}
where $\zeta_{k,ij}$ is between $\psi^T_k\hat\theta_{k-1,i}$ and $\psi^T_k\hat\theta_{k-1,j}$ such that $\hat{G}_{k,j}-\hat{G}_{k,i}= {g}_{k}(\zeta_{k,ij})(\tilde{\theta}_{k-1,j}^T\psi_k- \tilde{\theta}_{k-1,i}^T\psi_k)$. So, ${g}_{k}(\zeta_{k,ij})={g}_{k}(\zeta_{k,ji})$.

From  Lemma A.1 in \cite{KLZZ2024}, we have $\underline{g}=\inf_{k\geq 1}k^{\nu}{g}_{k}(x)>0$ for all $x\in[-\bar{\psi}\bar{\theta},\bar{\psi}\bar{\theta}]$, i.e.,
 \begin{align}\label{gk}
{g}_{k}(x)\geq \frac{\underline{g}}{k^{\nu}}, \forall x\in[-\bar{\psi}\bar{\theta},\bar{\psi}\bar{\theta}].
 \end{align}

Noticing $G(\cdot)$ is the distribution function of $Lap(0,1)$, the boundness of $\psi_k^T \hat{\theta}_{k-1,j}$ and $\hat{C}_k=\nu\ln k$,  we have
\vspace{-5pt}
\begin{align}\label{thmmc4}
&\mathbb{E}\big[(\hat{s}_{k,ij}-\hat{G}_{k,i})^2\big|\mathcal{F}_{k-1}\big]\nonumber\\
\leq&\mathbb{E}\big[\hat{s}_{k,ij}^2-2\hat{s}_{k,ij}\hat{G}_{k,i} +\hat{G}_{k,i}^2\big|\mathcal{F}_{k-1}\big]\nonumber\\
\leq&\hat{G}_{k,j}-2\hat{G}_{k,j}\hat{G}_{k,i}+\hat{G}_{k,i}^2\nonumber\\
=&O\(e^{-\nu\ln k}\)
=O\(\frac{1}{k^{\nu}}\).
\end{align}
From (\ref{gk}), $a_{ij}=a_{ji}$ and ${g}_{k}(\zeta_{k,ij})={g}_{k}(\zeta_{k,ji})$, we have
\begin{align}\label{thmmc5}
&\sum_{i=1}^{m}\tilde{\theta}_{k-1,i}^T\psi_k\sum_{j\in \mathcal{N}_{i}}a_{ij}{g}_{k}(\zeta_{k,ij})(\psi_k^T\tilde{\theta}_{k-1,j}- \psi_k^T\tilde{\theta}_{k-1,i})\nonumber\\
=&-\frac{1}{2}\sum_{i=1}^{m}\sum_{j=1}^{m}a_{ij}{g}_{k}(\zeta_{k,ij})\big((\psi_k^T\tilde{\theta}_{k-1,i})^2+(\psi_k^T\tilde{\theta}_{k-1,j})^2\nonumber\\
&\qquad\qquad\qquad\qquad\qquad\quad-2\tilde{\theta}_{k-1,i}^T\psi_k\psi_k^T\tilde{\theta}_{k-1,j}\big)\nonumber\\
\leq &- \frac{\underline{g}}{2k^{\nu}}\sum_{i=1}^{m}\sum_{j=1}^{m}a_{ij}(\psi_k^T\tilde{\theta}_{k-1,i}- \psi_k^T\tilde{\theta}_{k-1,j})^T\nonumber\\
&\qquad\qquad\qquad\qquad\quad\cdot(\psi_k^T\tilde{\theta}_{k-1,i}- \psi_k^T\tilde{\theta}_{k-1,j}).
\end{align}

For convenience of analysis, we introduce the notations
$\Phi_k=\diag\{\phi_{k,1},\ldots,\phi_{k,m}\}$,
$\hat{\Theta}_k=(\hat\theta_{k,1}^T,\ldots,\hat\theta_{k,m}^T)^T$,
$\tilde{\Theta}_k=(\tilde{\theta}_{k,1}^T,\ldots, \tilde{\theta}_{k,m}^T)^T$,
where $\diag\{\cdot,\ldots,\cdot\}$ denotes the block matrix formed in a diagonal manner of the corresponding vectors or matrices.

Based on (\ref{thmmc1})-(\ref{thmmc3}) and (\ref{thmmc4})-(\ref{thmmc5}), we can get
\vspace{-5pt}
\begin{align}\label{thmac1}
&\mathbb{E}[\tilde{\Theta}_{k}^T\tilde{\Theta}_{k}|\mathcal{F}_{k-h}]
=\mathbb{E}\big[\mathbb{E}[\tilde{\Theta}_{k}^T\tilde{\Theta}_{k}|\mathcal{F}_{k-1}]\big|\mathcal{F}_{k-h}\big]\nonumber\\
\leq &\mathbb{E}\big[\sum_{i=1}^m\tilde{\theta}_{k-1,i}^T\tilde{\theta}_{k-1,i}- \frac{2\beta\underline{f}}{k} \sum_{i=1}^m\tilde{\theta}_{k-1,i}^T\phi_{k,i}\phi^T_{k,i}\tilde\theta_{k-1,i} \nonumber\\
&- \frac{\alpha\underline{g}}{k}\sum_{i=1}^{m}\sum_{j=1}^{m}a_{ij}(\psi_k^T\tilde{\theta}_{k-1,i}- \psi_k^T\tilde{\theta}_{k-1,j})^T\nonumber\\
&\qquad\cdot(\psi_k^T\tilde{\theta}_{k-1,i}- \psi_k^T\tilde{\theta}_{k-1,j})\big|\mathcal{F}_{k-h}\big]+O\(\frac{1}{k^{2-\nu}}\)\nonumber\\
\leq & \mathbb{E}\big[\tilde{\Theta}_{k-1}^T\tilde{\Theta}_{k-1}
-\frac{2\beta \underline{f}}{k}\tilde{\Theta}_{k-1}^T\Phi_{k} \Phi^T_{k}\tilde\Theta_{k-1}\nonumber\\
&-\frac{2\alpha \underline{g}}{k} \tilde{\Theta}_{k-1}^T\mathcal{L}\otimes\psi_{k} \psi^T_{k}\tilde\Theta_{k-1}\big|\mathcal{F}_{k-h}\big]+O\(\frac{1}{k^{2-\nu}}\)\nonumber\\
%\leq&\tilde{\Theta}_{k-h}^T\tilde{\Theta}_{k-h}
%-\frac{2\beta \underline{f}}{k}\sum_{l=k-h}^{k-1}\tilde{\Theta}_{k-h}^T\Phi_{l+1} \Phi^T_{l+1}\tilde\Theta_{k-h}\nonumber\\
%&-\frac{2\alpha \underline{g}}{k} \sum_{l=k-h}^{k-1}\tilde{\Theta}_{k-h}^T\mathcal{L}\otimes\psi_{l+1} \psi^T_{l+1}\tilde\Theta_{k-h}+O\(\frac{1}{k^{2-\nu}}\)\nonumber\\
\leq&\tilde{\Theta}_{k-h}^T\tilde{\Theta}_{k-h}
-\frac{2}{k}\tilde{\Theta}_{k-h}^T\mathbb{E}\bigg[\beta \underline{f}\sum_{l=k-h}^{k-1}\Phi_{l+1} \Phi^T_{l+1}\nonumber\\
&-\alpha \underline{g}\mathcal{L}\otimes\psi_{l+1} \psi^T_{l+1}\bigg|\mathcal{F}_{k-h}\bigg]\tilde\Theta_{k-h} +O\(\frac{1}{k^{2-\nu}}\),
\end{align}
where the last inequality is by $\|\tilde{\theta}_{l,i}-\tilde{\theta}_{k-h,i}\|=O\(\frac{1}{k^{1-\nu}}\)$ for all $l=k-h,\ldots,k-1$, got by Assumptions \ref{AR}-\ref{AD} and (\ref{AE}).

From Assumption \ref{AR}, Lemma 4.1 in \cite{XG2018S}, the coding rule in Algorithm \ref{Algorithm}, we have
\begin{align}%\label{thmac2}
&-\(\beta \underline{f}\sum_{l=k-h}^{k-1}\Phi_{l+1} \Phi^T_{l+1}+\alpha \underline{g}\mathcal{L}\otimes \sum_{l=k-h}^{k-1}\psi_{l+1} \psi^T_{l+1}\)\nonumber\\
\leq&- \(\beta \underline{f}\sum_{l=k-h}^{k-1}\Phi_{l+1}\Phi^T_{l+1}
+ \alpha \underline{g}h\delta_{\psi}^2\mathcal{L}\otimes I_{n}\)\nonumber\\
\leq&-\frac{\sigma}{\beta \underline{f}\delta_{\phi}^2} \lambda_{\min}\(\beta \underline{f} \sum_{i=1}^m\sum_{l=k-h}^{k-1}\phi_{l+1,i}\phi^T_{l+1,i}\)I_{nm}\nonumber\\
\leq&-\sigma h, \nonumber
\end{align}
which together with (\ref{thmac1}) gives
\begin{align}\label{thmac3}
\mathbb{E}[\tilde{\Theta}_{k}^T\tilde{\Theta}_{k}|\mathcal{F}_{k-h}]
\leq&\(1-\frac{2\sigma h}{k}\)\tilde{\Theta}_{k-h}^T\tilde{\Theta}_{k-h}+O\(\frac{1}{k^{2-\nu}}\).
\end{align}
Then, by Lemma 1.3.2 in \cite{S_Guo2020} and $\sum_{k=1}^{\infty}\frac{1}{k^{2-\nu}}<\infty$, $\tilde{\Theta}_{k}^T\tilde{\Theta}_{k}$ converges almost surely to a bounded limit.
From (\ref{thmac3}),
\[
\mathbb{E}[\tilde{\Theta}_{k}^T\tilde{\Theta}_{k}]
\leq\(1-\frac{2\sigma h}{k}\)\mathbb{E}[\tilde{\Theta}_{k-h}^T\tilde{\Theta}_{k-h}]+O\(\frac{1}{k^{2-\nu}}\).
\]
Noting $\sum_{k=1}^\infty\frac{2\sigma h}{k}=\infty$ and $\lim_{k\rightarrow\infty}\frac{1}{k^{1-\nu}}=0$, we have $\lim_{k\rightarrow\infty}\mathbb{E}[\tilde{\Theta}_{k}^T\tilde{\Theta}_{k}]=0$ based on Lemma 5.4 in \cite{WGZZ2024}. Then, there is a subsequence of $\tilde{\Theta}_{k}^T\tilde{\Theta}_{k}$ that almost surely converges to 0.  Therefore, $\tilde\theta_{k,i}$ almost surely converges to 0.

\textbf{Part II}: Almost sure convergence rate.\\

Let $V_k=\frac{\(\left\lceil\frac{k}{h}\right\rceil+1\)^{ 1-\nu}} {\(\ln\(\left\lceil\frac{k}{h}\right\rceil+1\)\)^{1+\tau}} \tilde{\Theta}_{k}^T\tilde{\Theta}_{k}$. Then, by  (\ref{thmac3}), we have
\begin{align}%\label{thmac4}
&\mathbb{E}[V_{k}|\mathcal{F}_{k-h}]\nonumber\\
\leq&\(1-\frac{2\sigma h}{k}\)\frac{\(\left\lceil\frac{k}{h}\right\rceil+1\)^{ 1-\nu}\cdot\big(\ln\big(\left\lceil\frac{k-h}{h}\right\rceil+1\big)\big)^{1+\tau} } {\(\left\lceil\frac{k-h}{h}\right\rceil+1\)^{ 1-\nu}\cdot \big(\ln\(\left\lceil\frac{k}{h}\right\rceil+1\)\big)^{1+\tau}}V_{k-h}\nonumber\\
&+O\(\frac{1}{k(\ln k)^{1+\tau}}\)
\leq O\(\frac{1}{k(\ln k)^{1+\tau}}\)\nonumber\\
&+\(1-\frac{2\sigma }{\left\lceil\frac{k}{h}\right\rceil}\)\(1+\frac{ 1-\nu} {\left\lceil\frac{k}{h}\right\rceil}+O\(\frac{1}{k^2}\)\) V_{k-h}\nonumber\\
\leq&\(1+\frac{1-\nu-2\sigma}{\left\lceil\frac{k}{h}\right\rceil}+O\big(\frac{1}{k^2}\big)\) V_{k-h}+O\(\frac{1}{k(\ln k)^{1+\tau}}\)\nonumber\\
\leq&\(1+O\(\frac{1}{k^2}\)\) V_{k-h}+O\(\frac{1}{k(\ln k)^{1+\tau}}\),\nonumber
\end{align}
where the  last inequality is by use of $2\sigma>1-\nu$.
Noticing that $\sum_{k=1}^{\infty}\frac{1}{k^2}<\infty$ and
$\sum_{k=1}^{\infty}\frac{1}{k(\ln k)^{1+\tau}}<\infty$, we learn that $V_k$ converges almost surely a finite constant. Therefore, we have $\|\tilde\theta_{k,i}\|^2
=O\(\frac{(\ln k)^{1+\tau}}{k^{1-\nu}}\)$, a.s., for all $i=1,\ldots,m$.
\end{proof}

To describe the communication bit-rate, we first provide the following definition.
Given time interval $[1, k] \cap \mathbb{N}$, the global average communication bit-rate is
$$ \mathrm{\kappa}(k)=\frac{\sum_{l=1}^k\sum_{(i, j) \in \mathcal{E}}\kappa_{i j}(l)}{k\sum_{i=1}^{m}d_i} $$
where $\kappa_{i j}(l)$ is the bit number that the sensor $i$ sends to the sensor $j$ at time $l$, and $d_i$ is the neighbor number of the sensor $i$. Then, the following theorem shows the global communication bit-rate of the proposed algorithm decays to zero at a polynomial rate.

\begin{thm}\label{thm_gcr}
If Assumptions \ref{AG}-\ref{ADP} hold, then the global average communication bit rate $\kappa(k)$ satisfies
 \begin{align}%\label{thm_c1}
\mathrm{\kappa}(k)=O\left(\frac{1}{k^{\nu}}\right), \text{a.s.}\nonumber
 \end{align}
\end{thm}
\begin{proof}
From the defination of $\kappa_{i j}(l)$, we have
\begin{align}
& \mathbb{P}(\kappa_{i j}(l)=1)=G(\psi_l^T \hat{\theta}_{l-1,i}-\hat{C}_l)+G(-\psi_l^T \hat{\theta}_{l-1,i}-\hat{C}_l),  \nonumber\\
& \mathbb{P}(\kappa_{i j}(l)=0)=1-G(\psi_l^T \hat{\theta}_{l-1,i}-\hat{C}_l)-G(-\psi_l^T \hat{\theta}_{l-1,i}-\hat{C}_l),  \nonumber
\end{align}
and $\kappa_{i j}(l)\in\mathcal{F}_{l}$ for all $(i,j)\in\mathcal{E}$.
Then, based on the boundedness of $\psi_l^T \hat{\theta}_{l-1,i}$ and  Lemma A.1 in \cite{KLZZ2024}, we have
\begin{align}\label{Ezeta}
\mathbb{E}[\kappa_{i j}(l)|\mathcal{F}_{l-1}]
&=G(\psi_l^T \hat{\theta}_{l-1,i}-\hat{C}_l)+G(-\psi_l^T \hat{\theta}_{l-1,i}-\hat{C}_l)\nonumber\\
&=O\(G(-\hat{C}_l)\)=O\(\frac{1}{k^{\nu}}\),
\end{align}
which yields that
\begin{align}\label{sumEzeta}
& \sum_{l=1}^{k}\mathbb{E}[\kappa_{i j}(l)|\mathcal{F}_{l-1}]
=O\(k^{1-\nu}\).
\end{align}
From (\ref{Ezeta}) and $\kappa_{i j}(l)=0$ or $1$, we have
\begin{align}
& \mathbb{E}[\|\kappa_{i j}(l)-\mathbb{E}[\kappa_{i j}(l)|\mathcal{F}_{l-1}]\|^{\rho}|\mathcal{F}_{l-1}]\nonumber\\
\leq &\mathbb{E}[\|\kappa_{i j}(l)-\mathbb{E}[\kappa_{i j}(l)|\mathcal{F}_{l-1}]\|^{2}|\mathcal{F}_{l-1}]\nonumber\\
=&\mathbb{E}[\kappa_{i j}(l)|\mathcal{F}_{l-1}]-\(\mathbb{E}[\kappa_{i j}(l)|\mathcal{F}_{l-1}]\)^{2}
=O\(\frac{1}{k^{\nu}}\),  \nonumber
\end{align}
for all $\rho\geq 2$.
Then, from the weighted martingale difference sum estimation theorem in \cite{S_Guo2020}, we have
\begin{align}\label{sumMDEzeta}
& \sum_{l=1}^{k}(\kappa_{i j}(l)-\mathbb{E}[\kappa_{i j}(l)|\mathcal{F}_{l-1}])\nonumber\\
\leq &\sum_{l=1}^{k}\(\frac{1}{k^{\nu}}\)^{1/\rho}\frac{\kappa_{i j}(l)-\mathbb{E}[\kappa_{i j}(l)|\mathcal{F}_{l-1}]}{\(\frac{1}{k^{\nu}}\)^{1/\rho}}\nonumber\\
=&O\(s_k(\rho)\ln (s_k^{\rho}(\rho)+e)\)
=o\(k^{1-\nu}\), \text{ a.s. },
\end{align}
where $s_k(\rho)=\(\sum_{l=1}^{k}\((\frac{1}{k^{\nu}}\)^{1/\rho})^{\rho}\)^{1/\rho}
=o\(k^{1-\nu}\)$.

Then, by (\ref{sumEzeta}) and (\ref{sumMDEzeta}), we have
\[
\mathrm{\kappa}(k)=\frac{\sum_{l=1}^k\sum_{(i, j)\in \mathcal{E}}\kappa_{i j}(l)}{k\sum_{i=1}^{m}d_i}
=O\(\frac{1}{k^{\nu}}\), \text{ a.s. }
\]
\end{proof}

\begin{rema}\label{RTO}
Theorems \ref{thm_ac} and \ref{thm_gcr} reveal a trade-off between convergence and communication rates. Specifically, choosing $\alpha$ and $\beta$ such that $2\sigma \geq 1 - \nu$ yields a convergence rate of $O\left(\sqrt{\frac{\ln k}{k^{1-\nu}}}\right)$, where the convergence rate is inversely proportional to the communication bit rate. This trade-off offers practical guidance for tuning communication rates to meet desired convergence performance.
\end{rema}

\section{Numerical example}\label{sec:sim}
This section will illustrate the effectiveness of the proposed algorithm by a numerical example. Moreover, this example will also illustrate the joint effect of the sensors, i.e., the sensors in the network can achieve the estimation task that cannot be realized by any individual sensor only through decaying bit information exchanging between sensors.

\begin{wrapfigure}{r}{0.2\textwidth}
	\centering
	\includegraphics[width=0.18\textwidth]{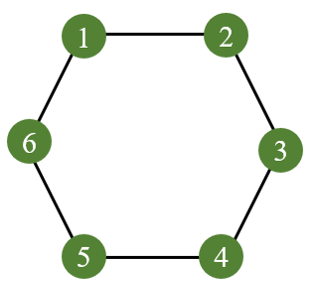}
	\caption{Communication graph.}
	\label{fig_graph}
\end{wrapfigure}

Consider a network composed of $m=6$ sensors, whose dynamics obeys the equation (\ref{M}) with $n=3$.
The communication graph $\mathcal{G}$ is set in Fig. \ref{fig_graph}.

In the dynamic model (\ref{M}), the unknown parameter $\theta=(1,-1,1)^{T}$, and its prior information is $\theta\in\Omega=[0,2]\times[-2,0]\times[0,2]$. The noise $d_{k,i}$ is i.i.d. Gaussian distribution $N(0,1)$.  Let the regressor $\phi_{k,i}$ be generated as
$\phi_{k,1}=\left(1-\frac{1}{3^k},0,0\right)$,
$\phi_{k,2}=\left(0,-1+\frac{1}{4^k},0\right)$,
$\phi_{k,3}=\left(0,0, 1-\frac{1}{2^k} \right)$,
$\phi_{k,4}=\left(-1+\frac{1}{2^k},0,0\right)$,
$\phi_{k,5}=\left(0,1-\frac{1}{2^k},0 \right)$,
$\phi_{k,6}=\left(0,0,-1+\frac{1}{5^k} \right)$.
One can verify that the regressor $\phi_{k,i}$ of the six sensors can cooperate to satisfy Assumption \ref{AR} with $h=1$. The packet loss $\gamma_{k,ij}^d$ follows Bernoulli distribution with $p=0.1$.

Then, we apply the proposed algorithm with the step size coefficients $\beta=70$ and $\alpha=20$ to give the estimate, where the linear coding rule $\psi_k$ in (\ref{Aen_z}) is sequential switching within the set $\{(1,0,0)^T, (0,1,0)^T,(0,0,1)^T\}$ and $\hat{C}_k=\nu\ln k$ in the event-triggered communication mechanism (\ref{Aga_e}) is set as $v=0,0.1,0.2,0.4,0.6$. And we repeat the simulation 100 times with the same initial values  $\hat{\theta}_{0,i}=(1/2,-1/2,1/2)$ to establish the empirical variance of estimation errors representing the mean square errors.

\begin{figure}[htbp]
	\centering
	\includegraphics[width=7.8cm]{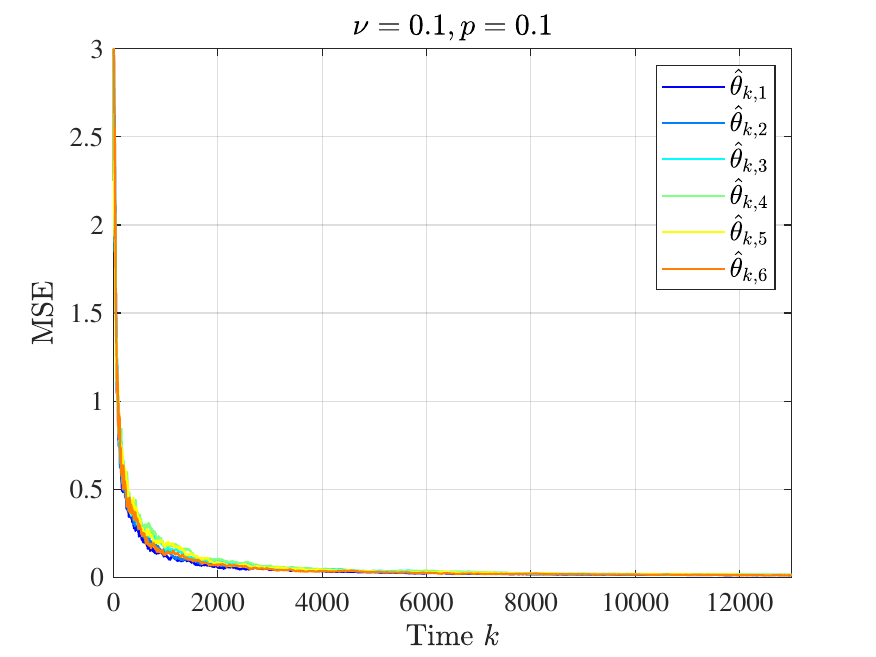}
	\caption{Convergence of Algorithm \ref{Algorithm} with $\nu=0.1$.}
	\label{fig_cp1}
\end{figure}

The mean square errors (MSE) of the proposed algorithm are shown in Fig. \ref{fig_cp1}, which shows the proposed algorithm can converge to the true parameter.
Fig. \ref{fig_cdbp1} shows the MSE trajectory comparison between Algorithm \ref{Algorithm} and its non-cooperative algorithm (i.e., $\hat{\theta}_{k,i}=\Pi_{\Omega}\{\hat{\theta}_{k-1,i}+\frac{\beta}{k} \phi_{k,i}(\hat{F}_{k,i}-s_{k,i})\}$). It shows the joint effect of the sensors under decaying communication bit-rate, i.e.,  the sensors in the network can achieve the estimation task that cannot be realized by any individual sensor only through a decaying bit-rate exchanging information between sensors.
\begin{figure}[htbp]
	\centering
	\includegraphics[width=7.8cm]{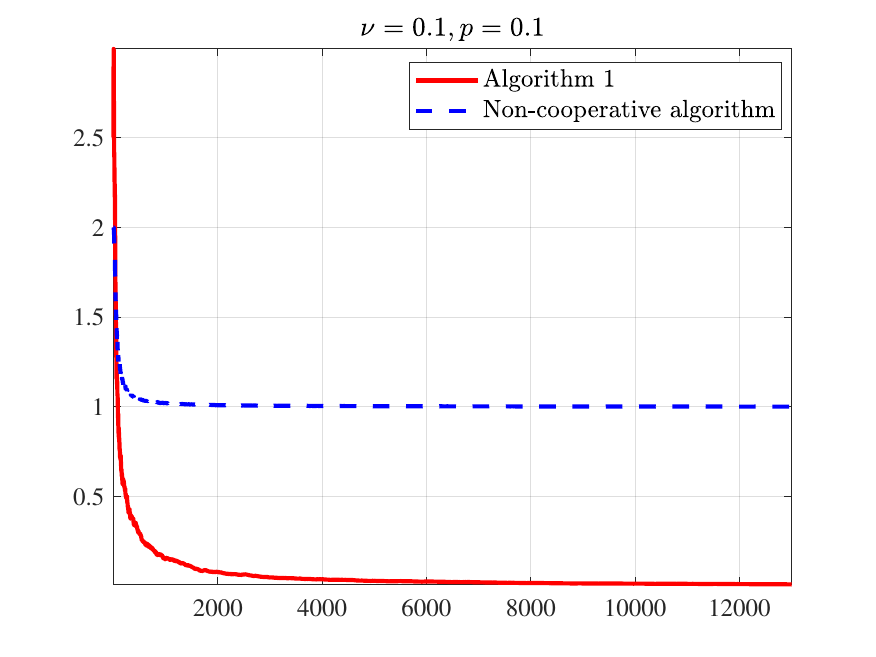}
	\caption{Comparison between Algorithm \ref{Algorithm} and non-cooperative algorithm.}
	\label{fig_cdbp1}
\end{figure}

In addition, Fig. \ref{fig_crp_cbrp} shows that the convergence rates and the global average communication bit-rate of the propose algorithm with different $\nu=0,0.1,0.2,0.4,0.6$. It demonstrates the trade-off between communication rates and convergence rates, which implies that a higher  communication rate leads to a faster convergence rate.

\begin{figure}[htbp]
	\centering
	\includegraphics[width=8.8cm]{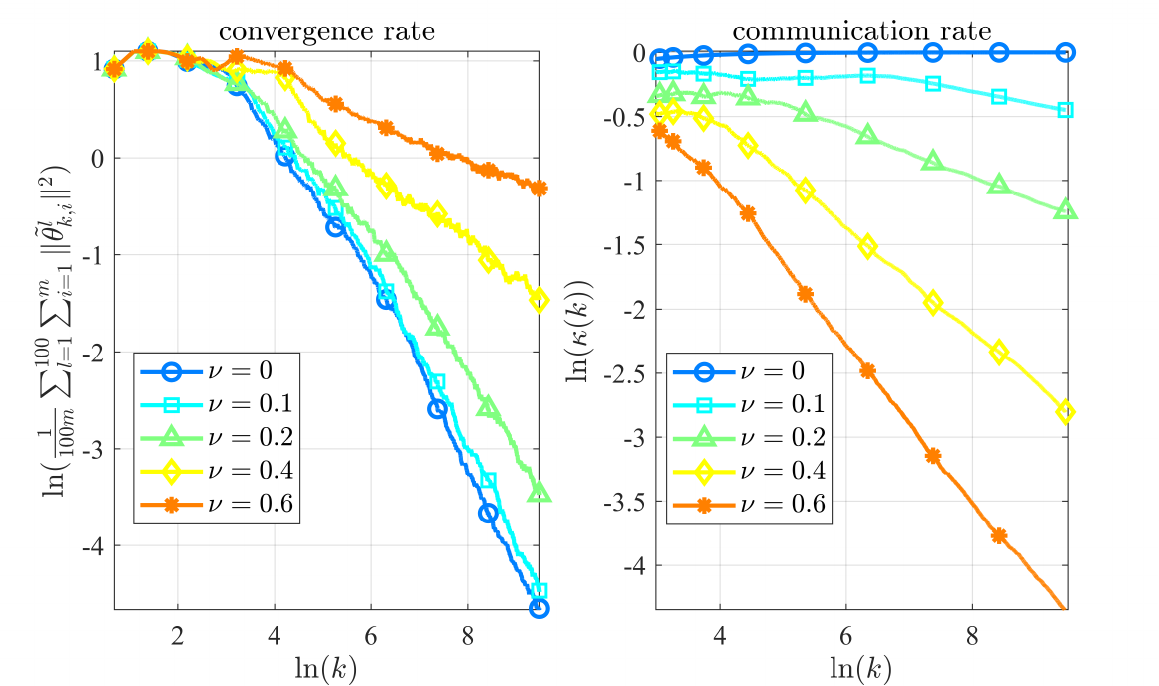}
	\caption{Convergence rates and average communication bitrates with different $\nu=0,0.1,0.2,0.4,0.6$.}
	\label{fig_crp_cbrp}
\end{figure}

\section{CONCLUSIONS}
This paper introduces a quantized distributed estimation algorithm  incorporating an event-triggered communication mechanism, capable of handling i.i.d. packet loss while relying only on a decaying global average communication bit-rate. Additionally, this paper establishes its almost sure convergence and analyzes its convergence rate. Furthermore, it explores the trade-off between communication rate and convergence rate,  offering guidance for communication design.
For future works, several intriguing directions can be explored, such as advanced event-triggered mechanism, and the interplay between quantized mechanism and event-triggered mechanisms.

\addtolength{\textheight}{-12cm}   % This command serves to balance the column lengths

%%%%%%%%%%%%%%%%%%%%%%%%%%%%%%%%%%%%%%%%%%%%%%%%%%%%%%%%%%%%%%%%%%%%%%%%%%%%%%%%

%%%%%%%%%%%%%%%%%%%%%%%%%%%%%%%%%%%%%%%%%%%%%%%%%%%%%%%%%%%%%%%%%%%%%%%%%%%%%%%%

%%%%%%%%%%%%%%%%%%%%%%%%%%%%%%%%%%%%%%%%%%%%%%%%%%%%%%%%%%%%%%%%%%%%%%%%%%%%%%%%
%\section*{APPENDIX}


%\section*{ACKNOWLEDGMENT}

%%%%%%%%%%%%%%%%%%%%%%%%%%%%%%%%%%%%%%%%%%%%%%%%%%%%%%%%%%%%%%%%%%%%%%%%%%%%%%%%

%\bibliographystyle{IEEEtran}
%\bibliography{mybib_ETQD_cdc.bib}

\end{document}